\documentclass[11pt,english]{article}
\usepackage{amsmath}
\usepackage{amsthm}
\usepackage{mathptmx}
\usepackage{newtxmath}
\usepackage[T1]{fontenc}
\usepackage[latin9]{inputenc}
\usepackage{geometry}
\usepackage{enumitem}
\usepackage{setspace}
\makeatletter
\IfFileExists{candleshape.def}{%
 \input{candleshape.def}}{}
\IfFileExists{dropshape.def}{%
 \input{dropshape.def}}{}
\IfFileExists{TeXshape.def}{%
 \input{TeXshape.def}}{}
\IfFileExists{triangleshapes.def}{%
 \input{triangleshapes.def}}{}

\newlength{\lyxlabelwidth}      
\@ifundefined{lettrine}{\usepackage{lettrine}}{}
\numberwithin{equation}{section}
\numberwithin{figure}{section}
\numberwithin{table}{section}
	\newenvironment{elabeling}[2][]%
	{\settowidth{\lyxlabelwidth}{#2}
		\begin{description}[font=\normalfont,style=sameline,
			leftmargin=\lyxlabelwidth,#1]}
	{\end{description}}
\theoremstyle{plain}
\ifx\thechapter\undefined
  \newtheorem{thm}{\protect\theoremname}
\else
  \newtheorem{thm}{\protect\theoremname}[chapter]
\fi
\theoremstyle{plain}
\ifx\thechapter\undefined
  \newtheorem{lem}{\protect\lemmaname}
\else
  \newtheorem{lem}{\protect\lemmaname}[chapter]
\fi
\ifx\proof\undefined\
  \newenvironment{proof}[1][\proofname]{\par
    \normalfont\topsep6\p@\@plus6\p@\relax
    \trivlist
    \itemindent\parindent
    \item[\hskip\labelsep
          \scshape
      #1]\ignorespaces
  }{%
    \endtrivlist\@endpefalse
  }
  \providecommand{\proofname}{Proof}
\fi
\theoremstyle{definition}
\ifx\thechapter\undefined
  \newtheorem{example}{\protect\examplename}
\else
  \newtheorem{example}{\protect\examplename}[chapter]
\fi

\makeatother

\usepackage{babel}
\providecommand{\examplename}{Example}
\providecommand{\lemmaname}{Lemma}
\providecommand{\theoremname}{Theorem}

\begin{document}
\title{The Over-and-Above Implementation of Reserve Policy in India\thanks{This paper subsumes and supersedes our earlier working paper titled
``\emph{Affirmative Action in India via Vertical, Horizontal, and
Overlapping Reservations: Comment.'' }We are grateful to\emph{ }Battal
Do\u{g}an, Isa Hafal\i r, and Bumin Yenmez for their feedback. We
acknowledge that we used ChatGPT to improve the readability and language
of the introduction section. }}
\date{December,  2023}
\author{Orhan Aygün\thanks{orhan.aygun@boun.edu.tr; University of Minnesota, Applied Economics
Department, 231 Ruttan Hall, 1994 Buford Ave, St. Paul, MN 55108,
USA; and Bo\u{g}aziçi University, Department of Economics, Natuk Birkan
Building, Bebek, Istanbul 34342, Turkey.} $\quad$and $\quad$Bertan Turhan\thanks{bertan@iastate.edu; Iowa State University, Department of Economics,
Heady Hall, 518 Farm House Lane, Ames, IA 50011, USA.}}
\maketitle
\begin{abstract}
The over-and-above choice rule is the prominent selection procedure
to implement affirmative action. In India, it is legally mandated
in allocating public school seats and government job positions. In
this paper, we present axiomatic characterization of the over-and-above
choice rule by rigorously stating policy goals as formal axioms. Moreover,
we characterize the deferred acceptance mechanism coupled with the
over-and-above choice rules for centralized marketplaces. \vfill{}
\end{abstract}
\begin{elabeling}{00.00.0000}
\item [{$\mathbf{Keywords:}$}] Market design, matching, reserve systems,
affirmative action. 
\end{elabeling}
$\mathbf{JEL\;Codes}$\emph{: }C78, D02, D47, D63, I21.\pagebreak{}

\section{Introduction}

The \textquotedbl\emph{Over-and-Above Choice Rule}\textquotedbl{}
has garnered significant attention in both theoretical constructs
of market design and practical applications for facilitating affirmative
action policies worldwide. When institutions earmark certain positions
for specific demographic groups, aiming to rectify socioeconomic disparities
in resource allocation, they prefer to allocate \emph{unreserved}
(or \emph{open-category}) positions before reserved category slots.
The foundational principle propounds that top-tier applicants from
reserved categories should first be considered for unreserved positions,
ensuring that some reserved category slots are allocated to candidates
who, in the absence of such affirmative action provisions, might not
secure admission. Notably, as Supreme Court decisions mandate it,
this rule has been operationalized in admissions to public universities
and the distribution of governmental positions in India.

Since its inception in 1950, India has been executing one of the most
extensive affirmative action programs globally to alleviate the historical
disparities faced by certain societal segments. This program is manifested
through a reservation system, entrenched in educational admissions
and government employment, governed by a robust legal framework, notably
influenced by Supreme Court rulings (Sönmez and Yenmez 2022; Aygün
and Turhan 2022). Under the federally mandated reserve policy, predefined
percentages of available positions in educational and governmental
institutions are set aside for applicants from specific socio-economic
groups. These percentages are distributed as 15\% for Scheduled Castes
(SC), 7.5\% for Scheduled Tribes (ST), 27\% for Other Backward Classes
(OBC), and 10\% for Economically Weaker Sections (EWS), collectively
termed \emph{vertical reserve categories}. Conversely, applicants
not falling under these classifications are categorized under the
\emph{General Category}. The policy mandates that individuals from
reserved categories must formally disclose their status to leverage
the benefits accorded by the program. The remainder, known as \emph{open-category}
(or \emph{unreserved}) positions, are accessible to all applicants,
including those from reserved categories. It is pertinent to note
that members of SC, ST, OBC, and EWS who opt not to declare their
affiliation to their respective reserved categories, alongside GC
members, are exclusively considered for open-category positions.

On top of the vertical reservations, so-called ``\emph{horizontal
reservations}\textquotedbl{} cut across all vertical categories in
allocating school seats and government job positions.\footnote{See the historic Supreme Court of India decision in Indra Sawhney
and others v. Union of India (1992) for the details of vertical and
horizontal reservations and their interactions. Sönmez and Yenmez
(2022) relate the mandates in this judgment to matching theory. } Horizontal reservations are typically provided for specific social
groups not defined by their caste or class but by other criteria,
such as gender (women) and physical disabilities. The key feature
of horizontal reservation is that the positions reserved under horizontal
reservation are filled first by eligible candidates from within each
vertical category.\footnote{Do\u{g}an and Y\i ld\i z (2023) call such choice functions ``\textbf{reserve-first}.\textquotedblright{}}
Horizontal reservations are implemented in each vertical category
via choice rules that fill horizontally reserved positions to eligible
applicants and then fill the rest of the positions by considering
all remaining applicants. For example, if there is a horizontal reservation
for women, women from each vertical category are first considered
to fill these reserved positions. Candidates from the respective vertical
category fill any unfilled seats in the horizontal quota. In other
words, the horizontal reservations are implemented as \emph{minimum
guarantees }(Sönmez and Yenmez, 2022).

The assignment of seats for college admissions and government positions
in India is fundamentally anchored in a \emph{merit-based} system.
Applicants in the same vertical reserve category are ranked based
on merit scores. This principle, called \emph{inter-se merit}, is
a key component of India's approach to balancing affirmative action
with the pursuit of excellence and meritocracy. As argued in Sönmez
and Yenmez (2022), the inter-se merit principle was mandated in the
historic Supreme Court of India (SCI) decision, \emph{Indra Sawhney
and others v. Union of India} (1992), and requires that given two
applicants who belong to the same category, if the one with a lower
merit-score is assigned a position, then the one with higher merit-score
must also be assigned a position. The Indian legal framework also
requires that when applicants from vertical reserve categories obtain
open-category positions, they are not counted against reservations
of their respective categories \emph{with the caveat that open-category
positions are to be allocated to higher-scoring applicants}. This
requirement is referred to as the \emph{over and above principle}.
Filling open-category positions before reserved category positions
in each institution achieves this principle. Moreover, in each vertical
category, all positions must be allocated as long as there are eligible
applicants for them. The following simple choice procedure, \emph{the
over-and-above choice rule,} has been used in India for decades. 
\begin{quotation}
\emph{First, applicants are selected for open-category positions one
at a time following the merit score ranking up to the capacity of
open-category. Then, applicants are selected for vertical reserve
categories following merit score ranking of applicants in each respective
category up to their capacities. }
\end{quotation}
In this paper, we characterize this widely used choice rule by rigorously
expressing the policy goals as formal axioms in the context of Indian
affirmative action. Furthermore, we characterize the celebrated deferred
acceptance (DA) algorithm of Gale and Shapley (1962) under the over-and-above
choice rules.\footnote{Do\u{g}an, Imamura, and Yenmez (2022) introduce a general theory for
the relationship between the characterizations of a choice rule and
the DA mechanism under this choice rule.}\footnote{Sönmez and Yenmez (2022) characterized a more complicated version
of the over-and-above-choice rule in a model with both vertical and
horizontal reservations. The authors present the characterization
of the ``\emph{Two-Step Meritorious Horizontal'' }choice rule with
four axioms: (1) maximal accommodation of HR protections, (2) no justified
envy, (3) non-wastefulness, and (4) compliance with VR protections.
Some of these axioms are already existing in the matching literature
and some of them are new ones that are proposed by the authors. In
our characterization, on the other hand, we only use policy goals
from the judgment in Indra Shawney (1992) by rigorously formalizing
them. } Sönmez and Yenmez (2022) attempted to ``\emph{verbally''} characterize
the over-and-above-choice rule in the absence of horizontal reservations
building on three principles from the Supreme Court judgment in Indra
Shawney (1992) without formally defining these principles as mathematical
axioms. In the appendix, we show that their characterization is erroneous,
given how they state the principles in writing. Our characterization
is based on principles obtained from the legal framework in India
as well. Given how prominent the over-and-above choice rule worldwide
is, its characterization is crucial to fully understand which principles
it implements. 

This paper contributes to the recently active market design literature
on Indian affirmative action. Echenique and Yenmez (2015) present
affirmative action in India as an example of controlled school choice.
Aygün and Turhan (2017) discuss issues in admissions to IITs. aswana
et al. (2018) designed and implemented a new joint seat allocation
process for the technical universities in India that has been implemented
since 2016. Aygün and Turhan (2020) formulate vertical reservations
and de-reservations in admissions to technical colleges in India and
introduce dynamic reserves choice rules as a possible solution. Aygün
and Turhan (2022) introduce the backward transfer choice rules and
the DA algorithm coupled with these rules, and show that it is the
unique desirable mechanism. Sönmez and Yenmez (2022) relates Indian
judiciary to matching theory and design a choice rule in a model with
vertical and horizontal reservations. 

This paper also contributes to the literature on reserve policies,
which started with the seminal papers of Hafal\i r, Yenmez, and Yildirim
(2013). Our paper contributes to the theoretical literature that studies
lexicographic choice functions. Important work in this strand of literature
include Chambers and Yenmez (2018), Do\u{g}an, Do\u{g}an and Yildiz
(2021), Do\u{g}an and Yildiz (2022), Avataneo and Turhan (2020), and
Aygün and Turhan (2020, 2022), among many others.

\section{The Preliminaries}

There is a finite set of institutions $\mathcal{S}=\left\{ s_{1},...,s_{m}\right\} $
and a finite set of individuals $\mathcal{I}=\left\{ i_{1},...,i_{n}\right\} $.
Institution $s$ has $\overline{q}_{s}$ positions. The vector $\left(q_{s}^{SC},q_{s}^{ST},q_{s}^{OBC},q_{s}^{EWS}\right)$
denotes the number of positions earmarked for SC, ST, OBC, and EWS
categories at institution $s$. We let $\mathcal{R}=\left\{ SC,ST,OBC,EWS\right\} $
to denote the set of reserve categories, and let $\mathcal{C}=\left\{ o,SC,ST,OBC,EWS\right\} $
to denote the set of all position categories, where $o$ denotes the
open-category. The number of open-category positions at institution
$s$ is $q_{s}^{o}=\overline{q}_{s}-q_{s}^{SC}-q_{s}^{ST}-q_{s}^{OBC}-q_{s}^{EWS}$.
The vector $q_{s}=\left(q_{s}^{o},q_{s}^{SC},q_{s}^{ST},q_{s}^{OBC},q_{s}^{EWS}\right)$
describes the \emph{distribution} of positions over categories at
institution $s$. 

The function $t:\mathcal{I}\rightarrow\mathcal{\mathcal{R}}\cup\left\{ GC\right\} $
denotes the category membership of individuals. For every individual
$i\in\mathcal{I}$, $t(i)$ denotes the category individual $i$ belongs
to. Reporting membership to SC, ST, and OBC is optional. Reserved
category members who do not report their membership are considered
GC applicants and eligible $\mathbf{only}$ for open-category positions.
Members of reserve category $r\in\mathcal{R}$ are eligible for both
open-category positions and reserved category-$r$ positions. We denote
a profile of reserved category membership by $T=\left(t_{i}\right)_{i\in\mathcal{I}}$,
and let $\mathcal{T}$ be the set of all possible reserved category
membership profiles. 

Merit scores induce strict\emph{ }meritorious ranking of individuals
at each institution $s$, denoted $\succ_{s}$, which is a linear
order over $\mathcal{I}\cup\{\emptyset\}$. $i\succ_{s}j$ means that
applicant $i$ has a higher priority (higher merit score) than applicant
$j$ at institution $s$. We write $i\succ_{s}\emptyset$ to say that
applicant $i$ is acceptable for institution $s$. Similarly, we write
$\emptyset\succ_{s}i$ to say that applicant $i$ is unacceptable
for institution $s$. The profile of institutions' priorities is denoted
$\succ=(\succ_{s_{1}},...,\succ_{s_{m}})$. 

For each institution $s\in\mathcal{S}$, the merit ordering for individuals
of type $r\in\mathcal{R}$, denoted by $\succ_{s}^{r}$, is obtained
from $\succ_{s}$ in a straightforward manner as follows: 
\begin{itemize}
\item for $i,j\in\mathcal{I}$ such that $t_{i}=r$, $t_{j}\neq r$, $i\succ_{s}\emptyset$,
and $j\succ_{s}\emptyset$, we have $i\succ_{s}^{r}\emptyset\succ_{s}^{r}j$,
where $\emptyset\succ_{t}^{r}j$ means individual $j$ is unacceptable
for category $r$ at institution $s$. 
\item for any other $i,j\in\mathcal{I}$, $i\succ_{s}^{r}j$ if and only
if $i\succ_{s}j$. 
\end{itemize}
For each institution $s\in\mathcal{S}$, its selection criterion is
summarized by a choice rule $C_{s}$. A choice rule $C_{s}$ simply
selects a subset from any given set of individuals. That is, for any
given $A\subseteq\mathcal{I}$, $C_{s}\left(A\right)\subseteq A$. 

\section{Characterization of the Over-and-above Choice Rule}

We now define three axioms that are legally required principles for\emph{
resource allocation problems in India in the absence of de-reservations}.
These axioms are concerned with institutions' choice functions. Given
a set of applicants $A\subseteq\mathcal{I}$, let $rank_{A}(i)$ be
the rank of applicant $i$ in set $A$ with respect to merit-ranking
$\succ_{s}$. That is, $rank_{A}(i)=k$ if and only if $\mid\left\{ j\in A\mid j\succ_{s}i\right\} \mid=k-1$.
Let $A^{r}=\left\{ i\in A\mid i\succ_{s}^{r}\emptyset\right\} $ be
the set of category $r\in\mathcal{R}$ eligible individuals in set
$A$. 

\paragraph{Over-and-above principle. }

Each individual $i\in A$ with $rank_{A}(i)\leq q_{s}^{o}$ must be
assigned to an open-category position. 

This property ensures that open-category positions obtained by reserve
category members are not counted against their reservations. By allocating
open-category positions to highest-scoring applicants, reserve category
positions can be given to members of reserve categories who would
not be able to receive these positions in the absence of reservation
policy. Moreover, over-and-above principle guarantees that cutoff
score of open-category is higher than cutoff scores of reserve categories.
Higher cut-off scores for the open category positions is a crucial
requirement in admissions to technical colleges in India. Baswana
et al. (2018) report that the authorities requested this property
to be satisfied. 

\paragraph{Within-category fairness.}

Given two individuals $i,j\in A$ such that $t_{i}=t_{j}$ and $i\succ_{s}j$,
if $j$ is assigned a position, then $i$ must also be assigned a
position. 

Within-category fairness requires that merit scores of applicants
are respected in each reserve category. That is, if a lower-scoring
applicant receives a position from category $c$, then a higher-scoring
applicant with the same category membership must receive a position
as well. 

\paragraph{Quota-filling subject to eligibility. }

If individual $i$ with $t_{i}=r$ is unassigned, then the number
of individuals who are matched to a category $r$ position must be
equal to $q_{s}^{r}$, for all $r\in\mathcal{R}$. 

We now give the formal description of the over-and-above choice rule. 

\subsubsection*{Over-and-Above Choice Rule $C_{s}^{OA}$}

Given an initial distribution of positions $q_{s}=\left(q_{s}^{o},q_{s}^{SC},q_{s}^{ST},q_{s}^{OBC},q_{s}^{EWS}\right)$,
a set of applicants $A\subseteq\mathcal{I}$, and a profile reserve
category membership $T=\left(t_{i}\right)_{i\in A}$ for the members
of $A$, the set of chosen applicants $C_{s}^{OA}(A,q_{s})$, is computed
as follows:

\paragraph{Stage 1. }

For open-category positions, individuals are selected following $\succ_{s}$
up to the capacity $q_{s}^{o}$. Let $C_{s}^{o}\left(A,q_{s}^{o}\right)$
be the set of chosen applicants. 

\paragraph{Stage 2. }

Among the remaining applicants $A^{'}=A\setminus C_{s}^{o}\left(A,q_{s}^{o}\right)$,
for each reserve category $t\in\mathcal{R}$, applicants are chosen
following $\succ_{s}^{t}$ up to the capacity $q_{s}^{t}$. Let $C_{s}^{t}\left(A^{'},q_{s}^{t}\right)$
be the set of chosen applicants for reserve category $t$.

Then, $C_{s}^{OA}(A,q_{s})$ is defined as the union of the set of
applicants chosen in stages 1 and 2. That is, 
\[
C_{s}^{OA}(A,q_{s})=C_{s}^{o}\left(A,q_{s}^{o}\right)\cup\underset{t\in\mathcal{R}}{\bigcup}C_{s}^{t}\left(A^{'},q_{s}^{t}\right).
\]

\begin{thm}
A choice rule satisfies the within-group fairness, over-and-above
principle, and capacity-filling subject to eligibility principles
if, and only if, it is the over-and-above choice rule $C_{s}^{OA}$. 
\end{thm}

\section{Deferred Acceptance Mechanism with Over-and-Above Choice Rules}

Admissions to technical universities is a centralized marketplace
with multiple institutions. An outcome in a centralized marketplace
is a matching. A \textbf{matching }$\mu$ is a function $\mu:\mathcal{I}\cup\mathcal{S}\rightarrow2^{\mathcal{I}}\cup\mathcal{S}\cup\left\{ \oslash\right\} $
such that 
\begin{itemize}
\item for any individual $i\in\mathcal{I}$, $\mu\left(i\right)\in\mathcal{S}\cup\{\oslash\}$, 
\item for any institution $s\in\mathcal{S}$, $\mu\left(s\right)\in2^{\mathcal{I}}$
such that $\mid\mu\left(s\right)\mid\leq\overline{q}_{s}$, 
\item for any individual $i\in\mathcal{I}$ and institution $s\in\mathcal{S}$,
$\mu\left(i\right)=s$ if and only if $i\in\mu\left(s\right)$. 
\end{itemize}
A matching specifies, for every institution, the set of individuals
assigned to that institution, but it does not specify categories under
which individuals are assigned. Associated with a matching is an \emph{$\mathbf{assignment}$}
which specifies a category each individual is accepted under in each
institution.\emph{ }Each individual's assignment is a $\mathbf{pair}$
of institution and category, and each institution's assignment is
a set of individual-category pairs. 

An \textbf{assignment }is a function $\eta:\;\mathcal{I}\cup\mathcal{S}\rightarrow\left(2^{\mathcal{I}}\cup\mathcal{S}\right)\times\mathcal{C}\bigcup\left\{ \oslash\right\} $
such that 
\begin{enumerate}
\item for each $i\in\mathcal{I}$, 
\[
\begin{cases}
\begin{array}{c}
\eta(i)\in\mathcal{S}\times\left\{ o\right\} \bigcup\left\{ \oslash\right\} \\
\eta(i)\in\mathcal{S}\times\left\{ t(i),o\right\} \bigcup\left\{ \oslash\right\} 
\end{array} & \begin{array}{c}
if\;t_{i}=GC,\\
if\;t_{i}\in\mathcal{R},
\end{array}\end{cases}
\]
\item for every $s\in\mathcal{S}$, $\eta(s)\subseteq2^{\mathcal{I}\times\mathcal{C}}$
such that $\mid\eta(s)\mid\leq\overline{q}_{s}$, and for all $r\in\mathcal{R}$,
\[
\mid\left\{ j\mid\left(j,r\right)\in\eta(s)\right\} \mid\leq q_{s}^{r},
\]
\item for every $\left(i,s\right)\in\mathcal{I}\times\mathcal{S}$, $\eta(i)=\left(s,c\right)$
if and only if $\left(i,c\right)\in\eta(i)$. 
\end{enumerate}
Let $\mu\left(\eta\right)$ be the matching induced by the assignment
$\eta$ and $\mu_{i}\left(\eta\right)$ be the institution that individual
$i$ is matched with. Similarly, $\mu_{s}\left(\eta\right)$ denotes
the set of individuals who are matched with institution $s$. Given
an assignment $\eta$, the matching $\mu\left(\eta\right)$ induced
by it is obtained as follows: 
\begin{itemize}
\item $\mu_{i}\left(\eta\right)=s$ if and only if $\eta(i)=\left(s,c\right)$
for some $c\in\mathcal{C}$, and 
\item $\mu_{s}(\eta)=\left\{ i\in\mathcal{I}\mid\left(i,c\right)\in\eta\left(s\right)\quad for\;some\quad c\in\mathcal{C}\right\} $.
\end{itemize}

\subsubsection*{Mechanisms }

A \textbf{mechanism}\emph{ }is a systematic way to map preference
and reserve category membership profiles of individuals to assignments,
given institutions' choice procedures. Technically, a mechanism $\varphi$
is a function $\varphi:\mathcal{P}\times\mathcal{T}\rightarrow\mathcal{M}$,
where $\mathcal{M}$ denotes the set of all assignments, given a profile
of institutional choice rules $\mathbf{C}=\left(C_{s}\right)_{s\in\mathcal{S}}$.
Note that the outcome of a mechanism is an assignment, not a matching.
It is an important modeling aspect because, in India, the outcomes
are announced as institution-category pairs for individuals. 

A mechanism is \textbf{incentive-compatible} if reporting the\emph{
true preference} and \emph{true reserve category membership} pair
is a weakly dominant strategy for each individual. Mathematically,
a mechanism $\mathcal{\varphi}$ is\emph{ }\textbf{incentive compatible}\emph{
}if for every profile $(P,T)\in\mathcal{P}\times\mathcal{T}$, and
for each individual $i\in\mathcal{I}$, there is no $\left(\widetilde{P}_{i},\widetilde{t}_{i}\right)$
such that 
\[
\mu_{i}\left[\varphi\left(\left(\widetilde{P}_{i},\widetilde{t}_{i}\right),\left(P_{-i},T_{-i}\right)\right)\right]P_{i}\mu_{i}\left[\varphi\left(P,T\right)\right].
\]

An assignment $\eta$ is \textbf{individually rational} if, for every
individual $i\in\mathcal{I}$, $\mu_{i}\left(\eta\right)R_{i}\emptyset.$
A mechanism $\varphi$ is if $\varphi\left(P,T\right)$ is\textbf{
individually rational} for any profile $\left(P,T\right)\in\mathcal{P}\times\mathcal{T}$. 

An assignment $\eta$ satisfies \textbf{within-category fairness}
if for every pair $\left(i,s\right)$ such that $sP_{i}\mu(i)$, for
all $j$ such that $\left(j,o\right)\in\eta\left(s\right)$ or $\left(j,t(i)\right)\in\eta(s)$
where $t(i)\in\mathcal{R}$,  we have $j\succ_{s}i$. A mechanism
$\varphi$ satisfies \textbf{within-category fairness} if $\varphi\left(P,T\right)$
satisfies within-category fairness for any profile $\left(P,T\right)\in\mathcal{P}\times\mathcal{T}$. 

An assignment $\eta$ is \textbf{non-wasteful} if for every pair $\left(i,s\right)$
such that $sP_{i}\mu(i)$, all positions $i$ is eligible at $s$
must be exhausted. That is, $\mid\left\{ j\mid\left(j,t(i)\right)\in\eta(s)\right\} \mid=q_{s}^{t(i)},$
for all $t(i)\in\mathcal{R}$. A mechanism $\varphi$ is \textbf{non-wasteful}
if $\varphi\left(P,T\right)$ is non-wasteful for any profile $\left(P,T\right)\in\mathcal{P}\times\mathcal{T}$. 

An assignment $\eta$ satisfies \textbf{over-and-above principle}
if, for all $i,j\in\mathcal{I}$ such that $\left(i,o\right)\in\eta(s)$
and $(j,r)\in\eta(s)$ where $r\in\mathcal{R}$, we have $i\succ_{s}j$.
A mechanism $\varphi$ satisfies the \textbf{over-and-above principle}
if $\varphi\left(P,T\right)$ satisfies open-first for any profile
$\left(P,T\right)\in\mathcal{P}\times\mathcal{T}$. 

The over-and-above principle aims to allocate open-category positions
to high-scoring applicants. Under every assignment that satisfies
this principle, open-category cutoff score is higher than the cutoff
scores of reserve categories. 

\subsubsection*{DA Mechanism under the Over-and-above Choice Rules (DA-OA)}

Let $P=\left(P_{i}\right)_{i\in\mathcal{I}}$ be the vector of the\emph{
reported} preference relations and $T=\left(t_{i}\right)_{i\in\mathcal{I}}$
be the  vector of reported reserve category membership of individuals.
Given institutions' priority rankings $\succ=\left(\succ_{s}\right)_{s\in\mathcal{S}}$
and the profile $\mathbf{q}=\left(q_{s}\right)_{s\in\mathcal{S}}$---therefore,
given $\left(C_{s}^{OA}\right)_{s\in\mathcal{S}}$---the outcome
of the DA-OA mechanism is computed as follows: 

\paragraph{Step 1. }

Each individual applies to his top choice institution. Let $\mathcal{A}_{s}^{1}$
be the set of individuals who apply to institution $s$, for each
$s\in\mathcal{S}$. Each institution $s$ holds  applicants in $C_{s}^{OA}\left(\mathcal{A}_{s}^{1},q_{s}\right)$
and rejects the rest. 

\paragraph{Step n$\protect\geq2$.}

Each individual who was rejected in the previous step applies to the
best institution that has not rejected him. Let $\mathcal{A}_{s}^{n}$
be the union of the set of individuals who were tentatively held by
institution $s$ at the end of Step $n-1$ and the set of new proposers
of $s$ in Step $n$. Each institution $s\in\mathcal{S}$ tentatively
holds individuals in $C_{s}^{OA}\left(\mathcal{A}_{s}^{n},q_{s}\right)$
and rejects the rest. 

The algorithm terminates when there are no rejections. 

We now present our characterization result for the DA-OA mechanism. 
\begin{thm}
Let $\left(\succ_{s}\right)_{s\in\mathcal{S}}$ be a profile of institutions'
priority orders. A mechanism $\varphi$ is individual rational, within-category
fair, non-wasteful, incentive-compatible, and satisfies the over-and-above
principle if, and only if, $\varphi$ is the DA-OA mechanism. 
\end{thm}

\section{APPENDIX }

\paragraph{Proof of Theorem 1. }

By definition of $C_{s}^{OA}$, the first $q_{s}^{o}$ top-ranked
applicants are assigned open-category positions. If there are fewer
than $q_{s}^{o}$ applicants, all of them are assigned open-category
positions. Thus, the over-and-above principle is satisfied. Because
each category fills its positions following the merit score rankings,
the within-category fairness is trivially satisfied. The only way
for an individual $i$ to be unassigned is that positions that $i$
is eligible for are filled with higher merit score individuals. Thus,
the capacity filling subject to eligibility is also satisfied.

Let $C$ be a choice rule that satisfies the over-and-above principle,
within-category fairness, and capacity-filling subject to eligibility.
By the over-and-above principle of $C$, we have $C^{o}\left(A\right)=C^{OA}\left(A\right)$.
If $\mid A\mid\leq q^{o}$, then we are done. Consider the case where
$\mid A\mid>q^{o}$. Capacity-filling subject to eligibility requires
that if an individual $i$ with $t\left(i\right)\in\mathcal{R}$ is
rejected, then $\mid C^{t\left(i\right)}\mid=q_{s}^{t\left(i\right)}$.
Capacity-filling subject to eligibility and within-category fairness
jointly imply that $C^{v}$ chooses the top $q_{s}^{v}$ individuals
from the set $\left(A\setminus C^{o}\left(A\right)\right)\bigcap A^{v}$,
for all $v\in\mathcal{R}$. Thus, $C\left(A\right)=C^{OA}\left(A\right)$. 

\paragraph{Proof of Theorem 2. }

We first show that the DA-OA mechanism satisfies all the axioms. Let
$\succ\equiv\left(\succ_{s}\right)_{s\in\mathcal{S}}$ and $C^{OA}=\left(C_{s}^{OA}\right)$
be the profile of institutions' priorities and over-and-above choice
rules. 

\paragraph{Individual rationality. }

In the course of the DA-OA algorithm, no individual $i\in\mathcal{I}$
proposes to an unacceptable institution. Thus, the DA-OA mechanism
is individually rational. 

\paragraph{Within-category Fairness. }

Given $\left(P,T\right)\in\mathcal{P}\times\mathcal{T}$, let $\eta$
be the DA-OA outcome and $\mu$ be the matching induced by $\eta$.
Consider $\left(i,s\right)\in\mathcal{I}\times\mathcal{S}$ such that
$sP_{i}\mu(i)$. Because $i$ is not matched with $s$, $i$ must
have applied to $s$ and got rejected during the DA-OA algorithm.
Let $n$ be the step at which $i$ gets rejected and $\mathcal{A}_{s}^{n}$
be the union of the set of individuals who were tentatively held by
$s$ at the end of step $\left(n-1\right)$ and the set of new proposers
to $s$ in step $n$. We have $i\in\mathcal{A}_{s}^{n}$ and $i\notin C_{s}^{OA}\left(\mathcal{A}_{s}^{n},q_{s}\right)$.
Because $i\notin C_{s}^{OA}\left(\mathcal{A}_{s}^{n},q_{s}\right)$,
for all $\left(j,o\right)\in\eta(s)$, we have $j\succ_{s}i$. If
$t(i)\in\mathcal{R}$, then for all $\left(j,t(i)\right)\in\eta(s)$
, we have $j\succ_{s}i$, by the definition of $C_{s}^{OA}$. Hence,
the DA-OA mechanism is within-category fair. 

\paragraph{Non-wastefulness. }

In the DA-OA algorithm, $i$ is rejected from open-category at $s$
if $q_{s}^{o}$ positions are exhausted in $C_{s}^{OA}$ in the step
of the DA-OA in which $i$ is rejected. Similarly, $i$ with $t\left(i\right)=r$
is rejected from category-$r$ at $s$ if she is rejected from open-category
in $C_{s}^{OA}$ and $q_{s}^{t(i)}$ category-$r$ positions are exhausted
in the iteration of $C_{s}^{OA}$, at the step of the DA in which
$i$ got rejected. Thus, the DA-OA mechanism is non-wasteful. 

\paragraph{Over-and-above Principle. }

Given $\left(P,T\right)\in\mathcal{P}\times\mathcal{T}$, let $n$
be the \emph{last }step of the DA-OA algorithm. Let $\mathcal{A}_{s}^{n}$
be the union of the set of individuals who were tentatively held by
institution $s$ at the end of step $\left(n-1\right)$ and the set
of new proposers to $s$ in step $n$. By definition of $C_{s}^{OA}$
choice rule, $C_{s}^{OA}\left(\mathcal{A}_{s}^{n},q_{s}\right)$ assigns
the top $q_{s}^{o}$ candidates to the open-category if $\mid\mathcal{A}_{s}^{n}\mid\geq q_{s}^{o}$.
Otherwise, it assigns all applicants to the open-category. Hence,
DA-OA satisfies the over-and-above principle. 

\paragraph{Incentive-compatibility. }

(i) We first prove that for a given $T=\left(t_{i}\right)_{i\in I}$,
the DA-OA cannot be manipulated via preference misreporting by showing
that $C_{s}^{OA}$ satisfies substitutability and size monotonicity.

\paragraph{Substitutability. }

Consider $i,j\in\mathcal{I}$ and $A\subset\mathcal{I}\setminus\{i,j\}$
such that $i\notin C_{s}^{BT}\left(A\cup\left\{ i\right\} \right)$.
We need to show that $i\notin C_{s}^{OA}\left(A\cup\left\{ i,j\right\} \right)$.
If $i\notin C_{s}^{OA}\left(A\cup\left\{ i\right\} ,q_{s}\right)$,
then $i$ is not in the top $q_{s}^{o}$ in the set $A\cup\{i\}$.
This implies that $i$ is not in top $q_{s}^{o}$ in the set $A\cup\{i,j\}$.
So, $i$ cannot be chosen for an open-category position from $A\cup\{i,j\}$. 

We now show that $i$ cannot be chosen for a reserve category $t_{i}\in\mathcal{R}$
position. First, suppose that $t_{i}=r$. Since $i$ was not chosen
for a category-r position from the set $A\cup\{i\}$, and when $j$
is added to the set $A\cup\{i\}$, $i$ cannot be chosen for a category-r
position because adding $j$ (weakly) increases the competition for
category-r positions. 

\paragraph{Size monotonicity. }

Consider $i\in\mathcal{I}$ and $A\subseteq\mathcal{I}\setminus\left\{ i\right\} $.
We need to show that $\mid C_{s}^{OA}\left(A\right)\mid\leq\mid C_{s}^{OA}\left(A\cup\left\{ i\right\} \right)\mid$.
If $i\notin C_{s}^{OA}\left(A\cup\left\{ i\right\} \right)$, then
we have $C_{s}^{OA}\left(A\cup\left\{ i\right\} \right)=C_{s}^{OA}\left(A\right)$
by the definition of $C_{s}^{OA}$. If $i\in C_{s}^{OA}\left(A\cup\left\{ i\right\} \right)$,
then either $C_{s}^{OA}\left(A\cup\left\{ i\right\} \right)=C_{s}^{OA}\left(A\right)\cup\left\{ i\right\} $
or $C_{s}^{OA}\left(A\cup\left\{ i\right\} \right)=C_{s}^{OA}\left(A\right)\cup\left\{ i\right\} \setminus\left\{ j\right\} $
for some $j\in C_{s}^{OA}\left(A\right)$ since each category fills
its quota subject to eligibility based on merit scores. 

Substitutability and size monotonicity of $C_{s}^{OA}$ for each $s$
imply strategy-proofness of the DA-OA by Hatfield and Milgrom (2005).
Thus, DA-OA cannot be manipulated via preference misreporting. 

\paragraph{(ii) }

Next, we prove that, for a given preference profile $P=\left(P_{i}\right)_{i\in I}$,
DA-OA cannot be manipulated by hiding reserve category membership.
This is implied by Theorem 3 of Aygün and Turhan (2022). 

From (i), we have $\mu_{i}[\varphi\left(P,T\right)]R_{i}\mu_{i}[\varphi\left((\widetilde{P}_{i},P_{-i}),T\right)].$
From part (ii), we obtain 
\[
\mu_{i}[\varphi\left((\widetilde{P}_{i},P_{-i}),T\right)]R_{i}\mu_{i}[\varphi\left((\widetilde{P}_{i},P_{-i}),(\widetilde{t}_{i},T_{-i})\right)].
\]
The transitivity of the relation $R_{i}$ gives $\mu_{i}[\varphi\left(P,T\right)]R_{i}\mu_{i}[\varphi\left((\widetilde{P}_{i},P_{-i}),(\widetilde{t}_{i},T_{-i})\right)].$

To show uniqueness, we first prove the following lemma. 
\begin{lem}
Let $\varphi$ be a mechanism that is individually rational, non-wasteful,
within-group fair, and satisfies the over-and-above principle. Let
$\left(P,T\right)\in\mathcal{P}\times\mathcal{T}$, $\varphi\left(P,T\right)=\eta$,
and $\mu=\mu\left(\eta\right)$. Then, $\mu$ is \emph{stable} with
respect to the profile $\mathbf{C}^{OA}=\left(C_{s}^{OA}\right)_{s\in\mathcal{S}}$.
That is, 

(1) for every individual $i\in\mathcal{I}$, $\mu_{i}\left(\eta\right)R\emptyset$, 

(2) for every institution $s\in\mathcal{S}$, $C_{s}^{OA}(\mu_{s},q_{s})=\mu_{s}$,
and 

(3) there is no $(i,s)$ such that $sP_{i}\mu_{i}$ and $i\in C_{s}^{OA}(\mu_{s}\cup\{i\},q_{s})$.
\end{lem}
\begin{proof}
Lemma 2 of Aygün and Turhan (2022) imply Lemma 1. 
\end{proof}
Incentive-compatibility of $\varphi$ implies that $\varphi$ is strategy-proof.
Lemma 2 states that $\varphi$ is a stable mechanism with respect
to $C=\left(C_{s}^{OA}\right)_{s\in\mathcal{S}}$. Because $C_{s}^{OA}$
is substitutable and size monotonic, by Hatfield, Kominers, and Westkamp
(2021), $C_{s}^{OA}$ satisfies observable substitutability, observable
size monotonicity, and non-manipulability via contractual terms conditions.
Then, by their Theorem 4 , DA-OA is the unique stable and strategy-proof
mechanism. Hence, $\varphi=DA-OA$. 

\subsection*{Informal Characterization in Sönmez and Yenmez (2022)}

In Section 2.2. of their paper, Sönmez and Yenmez (2022) present a
characterization of the over-and-above choice rule verbally as follows: 
\begin{quote}
\emph{In the absence of horizontal reservations, which will be introduced
in Section 2.3, the following three principles mandated in Indra Shawney
(1992) uniquely define a choice rule, thus making the implementation
of VR policies straightforward. First, an allocation must respect
inter se merit: Given two individuals from the same category, if the
lower merit-score individual is awarded a position, then the higher
merit-score individual must also be awarded a position}\textbf{\emph{.
}}\emph{Next,}\textbf{\emph{ VR protected positions must be allocated
on an ``over-and-above'' basis; that is, positions that can be received
without invoking the VR protections do not count against VR-protected
positions.}}\emph{ Finally, subject to eligibility requirements, all
positions have to be filled without contradicting the two principles
above. }\textbf{\emph{It is easy to see that these three principles
uniquely imply the following choice rule}}\emph{: First, individuals
with the highest merit scores are assigned the open-category positions.
Next, positions reserved for the reserve-eligible categories are assigned
to the remaining members of these categories, again based on their
merit scores. We refer to this choice rule as the over-and-above choice
rule.}
\end{quote}
The verbal characterization above is erroneous, given how the over-and-above
implementation of VR-protected positions is stated. As written, the
axiom does not specify that individuals assigned to open-category
positions must have higher scores than those assigned to reserved-category
positions, even though this may be the interpretation the authors
may have in mind, as remarks in this direction can be seen in their
introduction. However, without formally defining the axioms, both
directions of the characterization statement fail to hold.\footnote{More importantly, subsequent papers, such as Sönmez and Unver (2022),
state that the characterization of the over-and-above choice rule
in the absence of horizontal reservations was first formulated by
Sönmez and Yenmez (2022) but do not acknowledge the error in the characterization
as axioms are stated. The November 2022 version of Sönmez and Unver
(2022) formulates the ``compliance with the VR protections\textquotedbl{}
property in the absence of horizontal reservations and re-state the
characterization, but does not provide a proof. It can be accessed
at https://arxiv.org/pdf/2210.10166.pdf. Note that the characterization
statement in Proposition 0 of Sönmez and Unver (2022) is different
than the one in Sönmez and Yenmez (2022), given how the axioms are
stated in the latter. Sönmez and Unver (2022)'s re-formulation of
the axiom is identical to Aygün and Turhan (2022)'s formulation of
the \emph{over-and-above principle}.} 

As stated, the over-and-above implementation of VR protection property
is written as follows: For any set $A\subseteq\mathcal{I}$ and different
categories $r$ and $r^{'}$, $C^{r}\left(A\right)\cap C^{r^{'}}\left(A\right)=\emptyset$,
and the number of reserved category positions are independent of $C^{o}\left(A\right)$. 

Therefore, as written, this property does not require that \emph{positions
without invoking VR protections} must be acquired by candidates with
high merit scores. To see why this informal characterization does
not hold, consider the following examples.
\begin{example}
Consider an institution with two positions: one open category and
one category $r$. There are two individuals, $i$ and $j$. They
are both members of category $r$ and $i$ has a higher score than
$j$. According to $C^{OA}$, the high-scoring individual will be
chosen in the open category, and the low-scoring individual will be
chosen in category $r$. That is, 
\[
C^{OA}\left(\left\{ i,j\right\} \right)=\left(C^{o}\left(\left\{ i,j\right\} \right)=\left\{ i\right\} ,C^{r}\left(\left\{ i,j\right\} \right)=\left\{ j\right\} \right).
\]
\end{example}
Consider a choice rule $\widetilde{C}$ such that $\widetilde{C}\left(A\right)=C^{OA}\left(A\right)$
for all $A\in2^{\mathcal{I}}\setminus\left\{ i,j\right\} $. The rule
$\widetilde{C}$ is such that 
\[
\widetilde{C}\left(\left\{ i,j\right\} \right)=\left(C^{o}\left(\left\{ i,j\right\} \right)=\left\{ j\right\} ,C^{r}\left(\left\{ i,j\right\} \right)=\left\{ i\right\} \right).
\]
This choice rule satisfies the three axioms-{}-{}-as stated-{}-{}-in
Sönmez and Yenmez (2022). It trivially satisfies the first condition
because both individuals are awarded a position. It satisfies the
second axiom because $j$ being chosen in the open category does not
affect the number of positions in category $r$, so $i$ could be
chosen in category $r$. Finally, all positions are filled. 

The next example also shows that the ``if'' direction fails.
\begin{example}
Consider an institution with two positions: one open category and
one category $r$. There are two individuals, $i$ and $j$. Individual
$i$ is a member of category $r$ and $j$ is a general category member.
Suppose $i$ has a higher score than $j$. According to $C^{OA}$,
the high-scoring individual $i$ will be chosen in open-category,
and the low-scoring individual $j$ will be unassigned. That is, 
\[
C^{OA}\left(\left\{ i,j\right\} \right)=\left(C^{o}\left(\left\{ i,j\right\} \right)=\left\{ i\right\} ,C^{r}\left(\left\{ i,j\right\} \right)=\emptyset\right).
\]
 This choice rule violates the third axiom as stated in the following
sense: Assigning $j$ to open-category and $i$ to reserved category
$r$ positions complies with the first two axioms as stated and fills
both positions. Since both $i$ and $j$ receive positions, the first
axiom trivially holds. Assigning an individual to the open-category
does not change the number of available positions in category $r$.
Hence, the second axiom also holds. 
\end{example}
Thus, given how the axioms are stated, the over-and-above choice rule's
characterization with these axioms does not hold. The next example
shows that the stated axioms in Sönmez and Yenmez (2022) may induce
a choice rule that leaves the highest-scoring individual unassigned. 
\begin{example}
Suppose there are three individuals, $i$, $j$, and $k$. The institution
has two positions. One is an open-category, and the other is a reserved
category-$r$ position. The highest-scoring individual, $i$, is a
member of general category. Individuals $j$ and $k$ are members
of the reserved category $r$ and $j$ has a higher merit score than
$k$. Consider a choice rule $\widetilde{C}$ such that $\widetilde{C}\left(A\right)=C^{OA}\left(A\right)$
for all $A\in2^{\mathcal{I}}\setminus\left\{ i,j,k\right\} $. The
rule $\widetilde{C}$ is such that 
\[
\widetilde{C}\left(\left\{ i,j,k\right\} \right)=\left(C^{o}\left(\left\{ i,j,k\right\} \right)=\left\{ j\right\} ,C^{r}\left(\left\{ i,j,k\right\} \right)=\left\{ k\right\} \right).
\]
 This allocation satisfies all three axioms in Sönmez and Yenmez (2022).
Note that this allocation leaves the highest-scoring individual unassigned. 
\end{example}

\end{document}